\documentclass[11pt]{article}

\usepackage[margin=1in]{geometry}

\usepackage{amsmath,amssymb,amsthm}

\usepackage{tikz}
\usetikzlibrary{arrows}

\usepackage{mathtools}
\mathtoolsset{centercolon}

\usepackage[font=small]{caption}

\usepackage[colorlinks = true]{hyperref}

\hypersetup{
  pdftitle = {Interpolatability distinguishes LOCC from separable von Neumann measurements},
  pdfauthor = {Andrew M. Childs, Debbie Leung, Laura Mancinska, Maris Ozols}
}

\definecolor{darkred}  {rgb}{0.5,0,0}
\definecolor{darkblue} {rgb}{0,0,0.5}
\definecolor{darkgreen}{rgb}{0,0.5,0}
\hypersetup{
  urlcolor  = blue,         
  linkcolor = darkblue,     
  citecolor = darkgreen,    
  filecolor = darkred       
}

\newcommand{\Thm}[1]{\hyperref[thm:#1]{Theorem~\ref*{thm:#1}}}
\newcommand{\Lem}[1]{\hyperref[lem:#1]{Lemma~\ref*{lem:#1}}}
\newcommand{\Cor}[1]{\hyperref[cor:#1]{Corollary~\ref*{cor:#1}}}
\newcommand{\Def}[1]{\hyperref[def:#1]{Definition~\ref*{def:#1}}}
\newcommand{\Sect}[1]{\hyperref[sec:#1]{Section~\ref*{sec:#1}}}
\newcommand{\Fig}[1]{\hyperref[fig:#1]{Figure~\ref*{fig:#1}}}
\newcommand{\EqRef}[1]{\hyperref[eq:#1]{(\ref*{eq:#1})}}
\newcommand{\Eq}[1]{Eq.~\hyperref[eq:#1]{(\ref*{eq:#1})}}

\newcommand{\N}{\mathbb{N}}

\newcommand{\C}{\mathbb{C}}

\newcommand{\ket}[1]{|#1\rangle}
\newcommand{\bra}[1]{\langle#1|}

\newcommand{\ketbra}[1]{|#1\rangle\langle#1|}

\DeclarePairedDelimiter{\set}{\lbrace}{\rbrace}
\DeclarePairedDelimiter{\abs}{\lvert}{\rvert}

\newcommand{\mc}[1]{\mathcal{#1}}
\newcommand{\ct}{^{\dagger}}
\newcommand{\x}{\otimes}
\newcommand{\ve}{\varepsilon}
\newcommand{\id}{I}
\newcommand{\bip} [2]{\C^{#1}\x\C^{#2}}    

\DeclareMathOperator{\tr}{Tr}
\DeclareMathOperator{\spec}{spec}

\DeclareMathOperator{\Pos}{Pos}

\DeclareMathOperator{\SEP}{SEP}
\DeclareMathOperator{\LOCC}{LOCC}
\DeclareMathOperator{\LOCCN}{LOCC_{\N}}

\newtheorem{theorem}{Theorem}
\newtheorem{lemma}{Lemma}

\theoremstyle{definition}
\newtheorem{definition}{Definition}
\newtheorem*{example*}{Example}

\title{Interpolatability distinguishes LOCC \\ from separable von Neumann measurements}
\author{Andrew M.\ Childs,$^{1,2}$ Debbie Leung,$^{1,2}$ Laura Man\v{c}inska,$^{1,2}$ and Maris Ozols$^{1,2,3}$
\\[2pt] \normalsize
${}^1$ Department of Combinatorics \& Optimization, University of Waterloo \\ \normalsize
${}^2$ Institute for Quantum Computing, University of Waterloo \\ \normalsize
${}^3$ IBM TJ Watson Research Center}

\begin{document}
\maketitle

\begin{abstract}
Local operations with classical communication (LOCC) and separable operations are two classes of quantum operations that play key roles in the study of quantum entanglement. Separable operations are strictly more powerful than LOCC, but no simple explanation of this phenomenon is known. We show that, in the case of von Neumann measurements, the ability to interpolate measurements is an operational principle that sets apart LOCC and separable operations.
\end{abstract}

\section{Introduction} 
\label{sec:IIntro}

LOCC is defined operationally as the set of all quantum operations that several separated parties can implement given access to unlimited local quantum information processing and classical communication between them. Unfortunately, the class lacks a succinct mathematical description and hence is often hard to work with. In contrast, the class of separable operations, which is easily seen to encompass all LOCC operations, has a succinct and easy-to-use mathematical description. However, unlike LOCC, the class of separable operations does not have a natural operational interpretation.

Despite known quantitative separations between the two classes~\cite{IBM, Childs2012,Koashi07,Koashi09,ChitambarSep1,ChitambarSep2,Trines}, it is not understood what determines whether a given separable operation can or cannot be implemented with LOCC. In this paper, we draw intuition from the proof in \cite{IBM} and answer the above question for separable von~Neumann measurements. In \cite{IBM}, the authors divided any LOCC measurement into two stages but did not distinguish between LOCC and separable measurements in any other way. Here, we show that the possibility to interpolate a measurement to obtain partial information is intrinsic to LOCC but not separable von~Neumann measurements. More precisely, a separable von~Neumann measurement can be interpolated only if it can be decomposed into two nontrivial steps, the first of which can be performed by a finite LOCC protocol. Therefore, the ability to interpolate is an operational principle that distinguishes LOCC from separable von~Neumann measurements.

Another operational distinction between LOCC and separable measurements is suggested by the work of \cite{Koashi07,Koashi09} in the context of unambiguous state discrimination. Their result relies on the fact that LOCC protocols alternate between actions of the two parties, whereas general separable operations need not have this form. However, it remains open whether this property characterizes the difference between LOCC and separable operations, even in the setting of unambiguous state discrimination.

This paper is organized as follows. In \Sect{IMeasurements} we discuss separable and LOCC measurements in the context of state discrimination. In \Sect{IInterpolation} we define interpolation, the central concept of this work, and discuss 
interpolatability of arbitrary measurements (see \Thm{IKKB}). Our main result regarding interpolatability of separable and LOCC measurements is presented in \Sect{IMain} (see \Thm{IMain}). We conclude in \Sect{IDiscussion}.

For simplicity, we present our results for the bipartite case. However, as the reader can easily verify, the same arguments hold for any number of parties.

\section{Separable and LOCC measurements} 
\label{sec:IMeasurements}

Let $\Pos(\C^n)$ be the set of all positive semidefinite operators acting on $\C^n$ and let $[k] := \set{1, \dotsc, k}$. We describe a $k$-outcome measurement $\mathcal{M}$ on an $n$-dimensional system using its \emph{POVM elements} $\set{E_i}_{i\in[k]}$, where $\sum_{i=1}^k E_i = \id$ and each $E_i\in\Pos(\C^n)$. We call the set $\set{E_i}_{i\in[k]}$ the POVM of $\mc{M}$. The probability to obtain the outcome $i \in [k]$ upon measuring a state $\rho$ is $\tr(E_i \rho)$. 
If the POVM elements of $\mathcal{M}$ are mutually orthogonal projectors, we say that $\mathcal{M}$ is a \emph{projective measurement}. If in addition, each POVM element is rank one, then we say that $\mathcal{M}$ is a \emph{von~Neumann measurement}.
Such measurements are in one-to-one correspondence with ordered orthonormal bases of $\C^n$ (up to a phase factor for each basis vector). Therefore, a von~Neumann measurement can be specified by the orthonormal basis it measures in. We say that a measurement $\mathcal{M}$ is \emph{trivial} if all its POVM elements are proportional to the identity matrix. We use $\mathcal{I}$ to denote the trivial measurement with exactly one POVM element, $\id$.

\begin{definition}[Coarse graining]\label{def:ICoarse}
Let $\mathcal{M}$ and $\smash{\widetilde{\mathcal{M}}}$ be two measurements with POVMs $\set{E_1, \dotsc, E_k}$ and $\set{F_1, \dotsc, F_m}$, respectively. 
We say that $\smash{\widetilde{\mathcal{M}}}$ is a \emph{coarse graining} of $\mathcal{M}$ if there exists a partition $(\Lambda_1, \dotsc, \Lambda_m)$ of $[k]$ such that $F_i = \sum_{j \in \Lambda_i} E_j$ for all $i \in [m]$.
\end{definition}

We call a measurement $\mathcal{M}$ on $\bip{d_A}{d_B}$ \emph{separable}, and write $\mathcal{M}\in \SEP$, if each POVM element has the form $E_i = \sum_j a_j \x b_j$ for some $a_j \in \Pos(\C^{d_A})$ and $b_j \in \Pos(\C^{d_B})$.
Note that any separable measurement is a coarse graining of some measurement with product POVM elements.

Any LOCC protocol $\mathcal{P}$ implements a quantum operation of the form
\begin{equation}
 \rho \mapsto \sum_{m \in \Lambda}
 \ketbra{m} \otimes (A_m \otimes B_m) \rho (A_m\ct \otimes B_m\ct),
\end{equation}
where $\Lambda$ is the set of all terminating classical measurement records and $A_m \otimes B_m$ is the Kraus operator corresponding to record $m$ (see Section~2.2.2 of \cite{Childs2012} for more details). We refer to the operators $(A_m\ct A_m) \otimes (B_m\ct B_m)$ as the POVM elements of the protocol $\mathcal{P}$. We say that $\mathcal{P}$ implements a measurement $\mathcal{M}$ with POVM $\set{E_i}_i$ if the set $\Lambda$ can be partitioned into parts $\Lambda_i$ such that
\begin{equation}
  E_i = \sum_{j\in\Lambda_i} (A_j\ct A_j^{}) \otimes( B_j\ct B_j^{}).
\label{eq:IProj}
\end{equation}
Operationally, the partition corresponds to classical post-processing after the execution of $\mathcal{P}$, coarse-graining all outcomes in $\Lambda_i$ for each $i$. When $\mc{M}$ is a von~Neumann measurement, each POVM element $E_i$ in \Eq{IProj} is rank one. Thus $(A_j\ct A_j) \otimes (B_j\ct B_j)$ is proportional to $E_i$ for all $j \in \Lambda_i$, and $E_i$ is necessarily a tensor product operator. Hence, if a von~Neumann measurement in a basis $S$ can be implemented with LOCC then $S$ consists only of tensor product vectors. We call such bases \emph{product bases}. Since non-orthogonal states can never be perfectly distinguished, to discriminate states from an orthogonal set $S$ with certainty one can only apply non-disturbing measurements, defined as follows.

\begin{definition}[Non-disturbing measurement]
We say that a measurement $\mathcal{M}$ is \emph{non-disturbing} for a set of orthogonal states $S$ if $\bra{\psi} E \ket{\phi}=0$ for all POVM elements $E$ of $\mc{M}$ and all distinct $\ket{\psi},\ket{\phi}\in S$.
\label{def:INon-disturbing}
\end{definition}

Similar to \cite{Eric12}, we think of LOCC as a class of quantum operations rather than a class of tasks or protocols. Given a measurement $\mathcal{M}$, we write $\mathcal{M}\in \LOCCN$ if there exists a \emph{finite} LOCC protocol that implements $\mathcal{M}$.

We show in the following lemma that the problem of implementing a measurement $\mathcal{M}$ in a basis $S$ with finite LOCC is equivalent to the problem of perfectly discriminating the states from $S$ with finite LOCC. Throughout the paper, we use the two perspectives interchangeably.

\begin{lemma}\label{lem:IEquivalence}
Let $\mathcal{M}$ be a von~Neumann measurement in a basis $S$. A finite LOCC protocol $\mathcal{P}$ implements measurement $\mathcal{M}$ if and only if $\mathcal{P}$ discriminates the states from $S$ with certainty.
\end{lemma}

\begin{proof}
Clearly, given an LOCC protocol $\mathcal{P}$ that implements the
measurement $\mathcal{M}$ in the basis $S$, we can use the measurement
outcome to discriminate the states from $S$ with certainty using LOCC.
Conversely, suppose $\mathcal{P}$ discriminates the states from $S$
with certainty using LOCC. Consider any nonzero POVM element $E \otimes F$ of $\mathcal{P}$. Since $\mathcal{P}$ discriminates the states from $S$ with certainty, $\bra{\psi}(E\otimes F)\ket{\psi}=0$ for all but one of the states $\ket{\psi}\in S$. This means that $E\otimes F$ is proportional to one of the projectors in $\mathcal{M}$. By applying the same argument to all nonzero POVM elements of $\mathcal{P}$, we can partition them according to the states on which they project. Since the POVM elements must sum to $I$, the resulting LOCC protocol implements $\mathcal{M}$.
\end{proof}

\section{Interpolation of measurements} 
\label{sec:IInterpolation}

In this section we consider the problem of implementing a von~Neumann measurement in two stages, \emph{i.e.}, as a sequence of two measurements followed by coarse graining. In addition, we want to control how much progress is made during the first stage.

\subsection{Progress function}

To quantify the progress of the first measurement, we introduce a function that assigns numerical values to POVM elements. We take its range to be $[0,\infty)$ (the set of non-negative real numbers). Each value indicates how much progress is made when a particular measurement outcome occurs; a larger value corresponds to more progress.

Any such progress function must satisfy some operationally-motivated properties.
First, it must be continuous.
Second, it must vanish on POVM elements that are non-informative (\emph{i.e.}, proportional to the identity matrix).
Third, as we want to measure the progress \emph{conditioned} on having obtained a particular outcome, the progress function must be \emph{scale-invariant} (\emph{i.e.}, it remains the same when the POVM element is multiplied by a positive scalar).
Fourth, since coarse graining corresponds to discarding classical information, the progress achieved by a coarse-grained operator $\sum_i E_i$ must not exceed that of the most informative $E_i$.  We call the last condition \emph{quasiconvexity}.

\begin{definition}[Progress function]\label{def:IProgressFunction}
A continuous function $\mu: \Pos(\C^n) \setminus \set{0} \rightarrow [0,\infty)$ such that $\mu(\id) = 0$, $\mu(tE) = \mu(E)$ for all $t > 0$, and $\mu(E + F) \leq \max\set{\mu(E),\mu(F)}$ for all $E,F\in \Pos(\C^n) \setminus \set{0}$ is called a \emph{progress function}.
\end{definition}

\subsection{Interpolation}

We are interested in measurements $\mathcal{M}$ whose outcome statistics can be reproduced by a two-stage process: first perform some measurement $\mathcal{M}_1$ and then, conditioned on the outcome $i$, perform some other measurement \smash{$\mathcal{M}_2^{(i)}$}. More formally:

\begin{definition}[Composition of measurements]
Let $\mathcal{M}_1$ be a measurement with POVM $\set{E_1, \dotsc, E_k}$ and let $\mathcal{M}_2$ 
be a measurement with POVM $\smash{\set[\big]{\ketbra{i} \x E^{(i)}_j}_{ij}}$ 
such that  for each $i$, the set $\set{E^{(i)}_j}_j$ is a POVM of some measurement 
$\mathcal{M}_2^{(i)}$.
We say that a measurement $\mathcal{M}$ is
a \emph{composition} of $\mathcal{M}_1$ and $\mathcal{M}_2$, and write
$\mathcal{M} = \mathcal{M}_2 \circ \mathcal{M}_1$, if $\mathcal{M}$ is
a coarse graining (see \Def{ICoarse}) of a measurement with POVM 
\begin{equation}
  \set[\big]{E_i^{\frac{1}{2}} E^{(i)}_j E_i^{\frac{1}{2}}}_{ij}.
\label{eq:IComp}
\end{equation}
As a shorthand, we denote the second measurement $\mathcal{M}_2 = \bigoplus_{i \in [k]} \mathcal{M}_2^{(i)}$.
\end{definition}

Note that due to coarse graining, the POVM elements in \Eq{IComp} that
sum to a POVM element $E$ of the measurement $\mathcal{M}$ need not all be
proportional to $E$.
In such a case $\mathcal{M}_2\circ\mathcal{M}_1$ does not reproduce the post-measurement state of $\mathcal{M}$ for the outcome corresponding to $E$. However, if $\mathcal{M}$ is a von~Neumann measurement, then each POVM element $E$ of $\mathcal{M}$ is rank one, so the POVM elements in \Eq{IComp} that correspond to $E$ must be proportional to $E$. Therefore, any $\mathcal{M}_2\circ\mathcal{M}_1$ that reproduces the measurement statistics of a von~Neumann measurement $\mathcal{M}$ also reproduces its post-measurement states.

A progress function together with the ability to compose measurements allows us to speak of measurement interpolation, a two-stage implementation of a measurement where the amount of progress achieved in the first stage can be controlled.  
In general, some measurement outcomes might be more informative than others.  
In an $\ve$-interpolation, the progress after the first measurement is
at most $\ve$ regardless of the outcome obtained. 

\begin{definition}[$\ve$-interpolation]
\label{def:IInterpol}
Let $\ve\geq0$. An \emph{$\ve$-interpolation} of a measurement $\mathcal{M}$ with respect to a progress function $\mu$ is a pair of measurements $\mathcal{M}_1$ (with POVM $\set{E_1,\dotsc,E_k}$) and $\mathcal{M}_2$ such that
\begin{itemize}
\item $\displaystyle \max_{i\in[k]} \mu(E_i) = \ve$ and
\item $\mathcal{M}=\mathcal{M}_2\circ \mathcal{M}_1$ where $\mathcal{M}_2 = \bigoplus_{i\in[k]}\mathcal{M}_2^{(i)}$ for some measurements $\mathcal{M}_2^{(i)}$.
\end{itemize}
\end{definition}

The following theorem from \cite{KKB} (whose idea originates in \cite{IBM}) shows that any measurement can be $\ve$-interpolated. Note that this theorem does not require the progress function to be quasiconvex.

\begin{theorem}[\cite{KKB}]\label{thm:IKKB}
Let $\mu$ be any progress function (see \Def{IProgressFunction}). Then any measurement $\mathcal{M}$ 
can be $\ve$-interpolated with respect to $\mu$ for any $\ve\in [0,\lambda]$, where $\lambda:=\max_i \mu(F_i)$ and 
$\set{F_1,\dotsc,F_k}$ is the POVM for $\mathcal{M}$.
\end{theorem}

\begin{proof}
Let $c_1,\dotsc,c_k\geq0$ be constants and define $c:=(1+\sum_i c_i)^{-1}$.
Define POVM elements for $\mathcal{M}_1$ as 
\begin{equation}
  E_i := c \left(c_i \id + F_i\right)
\end{equation}
for all $i \in [k]$.
Let $\mc{M}_2 := \bigoplus_{i\in [k]}\mathcal{M}_2^{(i)}$, where $\mathcal{M}_2^{(i)}$ has POVM $\set{E^{(i)}_1,\dotsc,E^{(i)}_k}$ with 
\begin{align}
  E^{(i)}_j :=
    \begin{cases}
      \delta_{ij} \id 
      &\text{if $c_i=0$},\\
      c \left(c_i + \delta_{ij} \right) 
      E_i^{-\frac{1}{2}} 
      F_j
      E_i^{-\frac{1}{2}} 
      &\text{otherwise.}
 \label{eq:IEij}
 \end{cases}
\end{align}

First we check that the above definitions correspond to valid measurements.  
This is immediate for $\mathcal{M}_1$ and also for $\mathcal{M}_2^{(i)}$ when
$c_i=0$.  To see that each $\mathcal{M}_2^{(i)}$ is a valid measurement if
$c_i>0$, note that in this case the matrix $E_i$ has full rank and hence $E_i^{-\frac{1}{2}}$ is well-defined.  Furthermore, 
\begin{equation}
\sum_j E^{(i)}_j =  E_i^{-\frac{1}{2}} 
                    c\sum_j \left(c_i F_j + \delta_{ij} F_j\right)
                    E_i^{-\frac{1}{2}}
                 =  E_i^{-\frac{1}{2}} 
                    c\left(c_i \id + F_i\right)
                    E_i^{-\frac{1}{2}}
                 =  E_i^{-\frac{1}{2}} E_i E_i^{-\frac{1}{2}} = \id.
\end{equation}
We now show that the measurements $\mathcal{M}_1$ and $\mathcal{M}_2^{(i)}$ satisfy the two conditions of $\ve$-interpolation in \Def{IInterpol}.
First, note that $E_i$ changes continuously from $\tilde{F}_i := (1+\sum_{k \neq i} c_k)^{-1} F_i$ to $\id$ as $c_i$ changes from 0 to $\infty$. Since the progress function $\mu$ is continuous on nonzero operators, the parameter $c_i$ can be chosen so that $\mu(E_i)$ achieves any value between $\mu(\tilde{F}_i) = \mu(F_i)$ and $\mu(\id)=0$. Hence, for any $\ve\in[0,\lambda]$ we can choose $c_i$ so that $\mu(E_i) = \min\{\ve,\mu(F_i)\}$. Now recall that $\mu(t E_j)=\mu(E_j)$ for all $t>0$. Therefore, changing $c_i$ does not affect the value of $\mu(E_j)$ for $j \neq i$, so $\mu(E_i)$ can be adjusted independently for each $i$. Thus for any $\ve\in[0,\lambda]$ the parameters $c_i$ can be chosen so that $\max_i \mu(E_i)=\ve$.

Finally, to see that $\mathcal{M} = \bigl(\bigoplus_{i \in [k]} \mathcal{M}_2^{(i)}\bigr) \circ \mathcal{M}_1$, observe that the POVM elements of the right-hand side have the form
\begin{equation}
  E_i^{\frac{1}{2}} E^{(i)}_j E_i^{\frac{1}{2}}
  = c \left(c_i+\delta_{ij}\right) F_j.
\end{equation}
This expression holds even if $c_i=0$, since in that case $E_i^{\frac{1}{2}} E^{(i)}_j E_i^{\frac{1}{2}} = \delta_{ij} E_i = c \delta_{ij} F_j$. Since $\sum_i c \left(c_i+\delta_{ij}\right) F_j= F_j$, coarse graining over $i$ and labeling the measurement outcome by $j$ gives the desired measurement $\mathcal{M}$.
\end{proof}

\Thm{IKKB} states that any measurement can be $\ve$-interpolated for small enough $\ve$ when the type of measurement in the interpolation is unrestricted. In \Thm{IMain} we will see that this is not the case for interpolation with a restricted type of measurement.

\subsection{Interpolation in SEP}

When interpolating a separable measurement, it is natural to demand that both stages of the interpolation are also separable measurements.
Recall that if $\mc{M}_1$ is separable, then each POVM element $E$ for $\mc{M}_1$ must be of the form $E = \sum_j a_j \x b_j$. 
Here the coarse graining over index $j$ can be viewed as giving away the information about $j$ to the environment. We wish to measure the achieved progress by taking into account all extracted classical information, even if it is held by the environment. Therefore, when interpolating within $\SEP$ we modify \Def{IInterpol}.

\begin{definition}[$\ve$-interpolation in $\SEP$]
\label{def:ISEPInterpol}
Let $\mc{M}$ be a separable measurement.
We say that $\mc{M}$ can be \emph{$\ve$-interpolated in $\SEP$} for $\ve\geq 0$ if $\mc{M}$ has an $\ve$-interpolation $\big(\bigoplus_{i\in[k]}\mc{M}_2^{(i)}\big)\circ \mc{M}_1$ such that
\begin{itemize}
  \item the measurements $\mc{M}_1$ and $\mc{M}_2^{(i)}$ for all $i\in[k]$ are separable, and
  \item $\displaystyle \max_i \tilde{\mu}(E_i) = \ve$, where $\set{E_1,\dotsc,E_k}\subseteq \Pos(\bip{d_A}{d_B})$ is the POVM for $\mc{M}_1$ and $\tilde{\mu}$ is obtained by minimizing over all product decompositions:
  \begin{equation}
    \tilde{\mu}(E) := \min \set[\bigg]{\max_j \mu(a_j \x b_j) : E = \sum_j a_j \x b_j}
    \label{eq:ImuSEP}
  \end{equation}
  where $a_j \in \Pos(\C^{d_A})\setminus \set{0}$ and $b_j \in \Pos(\C^{d_B}) \setminus \set{0}$.
\end{itemize}
\end{definition}

Note that the minimum in the definition of $\tilde{\mu}$ is always achieved: $\mu$ is scale invariant and we can use Carath\'{e}odory's theorem to bound the number of terms in the sum $\sum_j a_j \x b_j$.

Consider the relationship between \Def{IInterpol} and \Def{ISEPInterpol}. Suppose we replace $a_j \x b_j$ with a general $F_j \in \Pos(\bip{d_A}{d_B})$ in \Eq{ImuSEP}. Then $E = F_1$ is a valid decomposition of $E$, so $\tilde{\mu}(E) \leq \mu(E)$. On the other hand, $\tilde{\mu}(E) \geq \mu(E)$ because $\mu(E) = \mu(\sum_j F_j) \leq \max_j \mu(F_j)$ by quasiconvexity of $\mu$. Therefore, the requirement
$\max_i \mu(E_i) = \ve$ in \Def{IInterpol} is equivalent to $\max_i
\tilde{\mu}(E_i) = \ve$ in \Def{ISEPInterpol} (without the product
constraint).

\subsection{Product interpolation}

We now define product interpolation, a notion that facilitates our derivations in the next section.

\begin{definition}[Product $\ve$-interpolation]
Let $\mc{M}_2 \circ \mc{M}_1$ be an $\ve$-interpolation of a measurement $\mc{M}$.
We say that $\mc{M}_2 \circ \mc{M}_1$ is a \emph{product $\ve$-interpolation} of $\mc{M}$ if all POVM elements of $\mc{M}_1$ have tensor product form.
\label{def:IProdInterpol}
\end{definition}

The following two simple lemmas show that product interpolation and interpolation in $\SEP$ are closely related. These lemmas are crucial for proving \Lem{IInterpol} and \Thm{IMain}, respectively.

\begin{lemma}
If $\mc{M}$ can be $\ve$-interpolated in $\SEP$ then $\mc{M}$ has a product $\ve$-interpolation that is also an interpolation in $\SEP$.
\label{lem:ISepToProd}
\end{lemma}

\begin{proof}
If $\mc{M}$ can be $\ve$-interpolated in $\SEP$, then the first-stage measurement $\mc{M}_1$ can be chosen to have tensor product POVM elements. This is because any POVM element $E = \sum_j a_j \x b_j$ of the first-stage measurement can be replaced with its fine-grained product operators $a_j \x b_j$ achieving the minimum in the definition of $\tilde{\mu}$.
\end{proof}

\begin{lemma}
Let $\mc{M}_2\circ\mc{M}_1$ be a product $\ve$-interpolation of $\mc{M}$ with $\mc{M}_2 \in \SEP$.
Then $\mc{M}_2\circ\mc{M}_1$ is also an $\ve$-interpolation of $\mc{M}$ in $\SEP$.
\label{lem:IProdToSep}
\end{lemma}

\begin{proof}
According to \Def{IProdInterpol}, $\mc{M}_1$ has only product POVM elements and hence is separable. By assumption, $\mc{M}_2 \in \SEP$ and therefore $\mc{M} = \mc{M}_2 \circ \mc{M}_1 \in \SEP$.

Let the POVM for $\mc{M}_1$ be 
$\set{c_1\otimes d_1,\dotsc, c_k\otimes d_k}$. According to \Def{IInterpol}, $\max_i \mu(c_i\otimes d_i) = \ve$. To see that $\mc{M}_2\circ\mc{M}_1$ is also an $\ve$-interpolation of $\mc{M}$ in $\SEP$, it remains to show that $\max_i \tilde{\mu}(c_i\otimes d_i) = \ve$. By the quasiconvexity of $\mu$, for any $c\otimes d$ and any decomposition $\sum_j a_j\otimes b_j = c\otimes d$, we have $\mu (c \otimes d) \leq \max_j \mu(a_j \otimes b_j)$. Therefore $\mu(c_i\otimes d_i) = \tilde{\mu}(c_i\otimes d_i)$ and the lemma follows.
\end{proof}

\section{Main result} 
\label{sec:IMain}

In this section we prove our main result concerning
$\ve$-interpolation of von~Neumann measurements in $\SEP$. Recall
from \Lem{IEquivalence} the equivalence between a von~Neumann
measurement and the task of state discrimination for an orthonormal
basis $S$.
If $\bra{\psi} E \ket{\psi} = 0$ for some $\ket{\psi} \in S$ and some
POVM element $E$, the corresponding outcome eliminates $\ket{\psi}$.
To capture the intuition that significant progress is made in this
case, we focus on progress functions whose values for such $E$ cannot be  arbitrarily small.

\begin{definition}[Threshold]\label{def:IThreshold}
A progress function $\mu$ (see \Def{IProgressFunction}) has
\emph{threshold} $\mu_0 > 0$ with respect to an orthonormal basis $S$
if $\mu(E) \geq \mu_0$ for all nonzero $E \in \Pos(\C^n)$ such that $\bra{\psi} E \ket{\psi} = 0$ for some $\ket{\psi} \in S$.
\end{definition}

As a concrete example, consider the following progress function \cite{Childs2012}.

\begin{example*}
Let $S\subseteq \C^n$ be an orthonormal basis. Consider $\mu:\Pos(\C^n) \setminus \{0\} \to [0,\infty)$ given by
\begin{equation}
  \mu(E) :=
  \frac{ \max_{\ket{\psi} \in S} \bra{\psi} E \ket{\psi}}{\tr(E)}
   - \frac{1}{\abs{S}}.
  \label{eq:Imu}
\end{equation}
The first term in \Eq{Imu} is the maximum probability of making a
correct guess if the outcome corresponds to $E$, so $\mu$ measures the
deviation of the best guess from a uniformly random guess. 
It is easy to verify that $\mu$ satisfies the conditions of
\Def{IProgressFunction} and hence is a valid progress function.
If $\bra{\psi} E \ket{\psi} = 0$ for some $\ket{\psi} \in S$, then the
first term is at least $\frac{1}{n-1}$, so $\mu$ has threshold $\mu_0 =
\frac{1}{n(n-1)}$.
\end{example*}

The following lemma shows that if a separable von~Neumann measurement $\mathcal{M}$ in a basis $S$ can be $\ve$-interpolated for some small $\ve$, then there exists a nontrivial {\em local} measurement that is non-disturbing for $S$ (see \Def{INon-disturbing}). Intuitively this means that some part of the measurement $\mathcal{M}$ can be implemented by LOCC (we formalize this intuition later in \Thm{IMain}, our main result).

\begin{lemma}\label{lem:IInterpol}
Let $\mc{M} \in \SEP$ be a von~Neumann measurement in a basis $S \subseteq \bip{d_A}{d_B}$ and let $\mu$ be a progress function with threshold $\mu_0$ with respect to $S$ (see \Def{IThreshold}). If $\mc{M}$ can be $\ve$-interpolated in $\SEP$ for some $\ve \in (0, \mu_0)$, then there exists a projective measurement $\mathcal{L}$ of the form $\mc{A} \x \mc{I}$ or $\mc{I} \x \mc{B}$ that is non-disturbing for $S$ and achieves progress $\mu(E) \geq \mu_0$ for all $E\in\mathcal{L}$.
\end{lemma}

\begin{proof}
Assume that $\mathcal{M}$ admits an $\ve$-interpolation in $\SEP$ for some $\ve \in (0,\mu_0)$. By \Lem{ISepToProd} it also admits a product $\ve$-interpolation (see \Def{IProdInterpol}). Let $\mathcal{M}_1$, with POVM $\set{E_i = a_i\otimes b_i}_i$, be the first measurement in some product $\ve$-interpolation of $\mc{M}$.
Since the measurement $\mathcal{M}$ perfectly discriminates the states from $S$, 
$\mathcal{M}_1$ must be non-disturbing, \emph{i.e.},
\begin{equation}
  \bra{\psi_j} E_i \ket{\psi_k} = 0
\end{equation}
for all $E_i$ and all distinct $j,k \in[d_A d_B]$. It follows that each $E_i$ is diagonal in the basis $S$. Thus, for each $i$ and $k$ there exists $\lambda_{ik} \geq 0$ such that
\begin{equation}
  E_i \ket{\psi_k} = \lambda_{ik} \ket{\psi_k}.
  \label{eq:IEigenvec}
\end{equation}
If any $\lambda_{ik}=0$, then $\bra{\psi_k} E_i \ket{\psi_k}=0$ and hence $\mu(E_i)\geq \mu_0$. Yet this contradicts the interpolation condition requiring that $\mu(E_i) \leq \ve < \mu_0$. Thus $\lambda_{ik} > 0$ for all $i,k$.

Now, using the fact that each $E_i = a_i \otimes b_i$ and each $\ket{\psi_k} = \ket{\alpha_k} \otimes \ket{\beta_k}$ for a product basis $S$, we rewrite \Eq{IEigenvec} as
\begin{equation}
  (a_i \otimes b_i) \ket{\alpha_k} \otimes \ket{\beta_k}
  = \lambda_{ik} \ket{\alpha_k} \otimes \ket{\beta_k}.
\end{equation}
Strict positivity of $\lambda_{ik}$ implies $b_i \ket{\beta_k} \neq 0$ and thus
\begin{equation}
  a_i \ket{\alpha_k} = \eta_{ik} \ket{\alpha_k},
\end{equation}
where $\eta_{ik} = \lambda_{ik} / \| b_i \ket{\beta_k} \|_2 > 0$.  Thus $(a_i \x \id_B) \ket{\psi_k} = a_i \ket{\alpha_k} \x \id_B \ket{\beta_k} = \eta_{ik} \ket{\psi_k}$, so
\begin{equation}
  \bra{\psi_j} (a_i \otimes \id_B) \ket{\psi_k} = 0
\end{equation}
for all distinct $j, k$.  Thus the matrix $a_i \otimes \id_B$ is diagonal in the basis $S$, and so is each $\Pi_{i,\eta} \otimes \id_B$, where $\Pi_{i,\eta}$ is the projector onto the eigenspace of $a_i$ with eigenvalue $\eta$.  Hence
\begin{equation}
  \bra{\psi_j} (\Pi_{i,\eta} \x \id_B) \ket{\psi_k} = 0
  \label{eq:IJoint}
\end{equation}
for all distinct $j,k$ and all $\eta \in \spec(a_i)$. 
If $\mc{A}_i$ is the projective measurement onto the eigenspaces of
$a_i$, with POVM $\set{\Pi_{i,\eta} : \eta \in \spec(a_i)}$, 
then the
joint measurement $\mc{A}_i \x \mc{I}$ is non-disturbing for $S$
according to \Eq{IJoint}. Note that unless $a_i = \id_A$, we have
$\bra{\psi} (\Pi_{i,\eta} \x \id_B) \ket{\psi} = 0$ for some
$\ket{\psi}\in S$. In this case $\mu(E) \geq \mu_0$ for all POVM 
elements $E$ 
of $\mc{A}_i \otimes \mc{\id}$. The same holds for $\mc{I} \x \mc{B}_i$
which can be defined similarly.

It remains to show that for some $i$ at least one of $\mc{A}_i \x \mc{I}$ and $\mc{I} \x \mc{B}_i$ is nontrivial. Consider an $i$ such that $\mu(E_i) = \ve$.  Since $\ve > 0$, $a_i \x b_i $ is not proportional to the identity matrix. Thus either $a_i$ is not proportional to the identity matrix and hence $\mc{A}_i \x \mc{I}$ is nontrivial, or $b_i$ is not proportional to the identity matrix and $\mc{I} \x \mc{B}_i$ is nontrivial. 
\end{proof}

Now we are ready to prove our main theorem, establishing interpolatability as an operational principle that distinguishes LOCC and separable von~Neumann measurements.

\begin{theorem}\label{thm:IMain}
Let $\mc{M} \in \SEP$ be a von~Neumann measurement in a basis $S \subseteq \bip{d_A}{d_B}$ and let $\mu$ be a progress function with threshold $\mu_0$ with respect to $S$ (see \Def{IThreshold}). Then $\mc{M}$ can be $\ve$-interpolated in $\SEP$ for some $\ve \in (0, \mu_0)$ if and only if $\mc{M} = \mc{M}_2 \circ \mc{M}_1$ for some $\mc{M}_2\in\SEP$ and some $\mc{M}_1 \in \LOCCN$ that achieves progress $\mu(E)\geq\mu_0$ for all POVM elements $E$ of $\mc{M}_1$.
\end{theorem}

\begin{proof}
($\Rightarrow$) Assume that $\mc{M}$ can be $\ve$-interpolated in $\SEP$ for some $\ve \in (0, \mu_0)$. Then by \Lem{IInterpol} there exists a local $k$-outcome measurement $\mc{A}$ on one of the parties, say Alice, such that $\mc{A} \x \mc{I}$ is non-disturbing for $S$ and achieves progress $\mu(E) \geq \mu_0$ for all POVM elements $E$ of $\mc{A} \x \mc{I}$. Choose $\mc{M}_1 = \mc{A} \x \mc{I}$ and $\mc{M}_2 = \bigoplus_{i \in [k]} \mc{M}$. Since $\mc{A} \x \mc{I}$ is non-disturbing for $S$, coarse graining according to the outcomes of $\mc{M}_2$ implements the original measurement $\mc{M}$ in the basis $S$. Hence $\mc{M} = \mc{M}_2 \circ \mc{M}_1$ where $\mc{M}_2 \in \SEP$ and $\mc{M}_1 \in \LOCCN$.

($\Leftarrow$) Assume that $\mc{M} = \mc{M}_2 \circ \mc{M}_1$ for some
$\mc{M}_2\in\SEP$ and some $\mc{M}_1 \in \LOCCN$ that achieves
progress $\mu(E) \geq \mu_0$ for all $E$ in the POVM of $\mc{M}_1$.
To obtain the desired $\ve$-interpolation of $\mc{M}$, we locate the
earliest measurement in an LOCC implementation for
$\mc{M}_1$ that achieves nonzero progress.  By $\ve$-interpolating
this local measurement we obtain an $\ve$-interpolation of $\mc{M}_1$
and hence of $\mc{M}$. We now formalize this idea.

Consider an LOCC protocol for implementing $\mathcal{M}_1$. We can naturally represent this protocol as a rooted tree $\mc{T}$, where the nodes in each level correspond to measurements  performed in the corresponding round of the protocol (see Section~2.2.4 of \cite{Childs2012} for more explanation). We define a subtree $\mc{T}'$ of $\mc{T}$ recursively as follows (see \Fig{ITree} for an example). First, we include the root of $\mc{T}$ in $\mc{T}'$. Next, if a vertex $v$ is in $\mc{T}'$ and all children of $v$ have zero progress, then we include the children of $v$ in $\mc{T}'$ as well. 
We obtain the desired $\ve$-interpolation of $\mc{M}$ by interpolating the measurement at some leaf $v'$ of $\mc{T}'$.

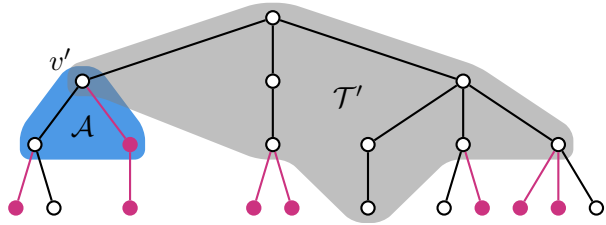
\begin{figure}[!ht]
\centering

\pgfdeclarelayer{background}
\pgfsetlayers{background,main}

\definecolor{Blk}{rgb}{0.0,0.0,0.0}
\definecolor{Red}{rgb}{0.8,0.2,0.5}
\definecolor{Blu}{rgb}{0.3,0.6,0.9}

\begin{tikzpicture}[
  thick,
  circ/.style = {circle, draw = Blk, fill = white, inner sep = 0mm, minimum size = 5pt},
  dot/.style = {circle, draw = Red, fill = Red, inner sep = 0mm, minimum size = 5pt},
  de/.style = {Red},
  thn/.style = {thick},
  > = latex']

\def\dx{18}
\def\dy{24}
\tikzstyle{level 1} = [level distance = \dy, sibling distance = 4*\dx]
\tikzstyle{level 2} = [level distance = \dy, sibling distance = 2*\dx]
\tikzstyle{level 3} = [level distance = \dy, sibling distance = 0.8*\dx]

\path [Blk]
node (r) at (5,0) [circ] {}
child {
  node (a) [circ] {}
  child [thn] {
    node (L) [circ] {}
    child [de] {node [dot] {}}
    child {node [circ] {}}
  }
  child [de,thn] {
    node (R) [dot] {}
    child [de] {node [dot] {}}
  }
}
child {
  node (T') [circ] {}
  child {
    node (b) [circ] {}
    child [de,thn] {node [dot] {}}
    child [de,thn] {node [dot] {}}
  }
}
child {
  node (f) [circ] {}
  child {
    node [circ] {}
    child {node (c) [circ] {}}
  }
  child {
    node (d) [circ] {}
    child [thn] {node [circ] {}}
    child [de,thn] {node [dot] {}}
  }
  child {
    node (e) [circ] {}
    child [de,thn] {node [dot] {}}
    child [de,thn] {node [dot] {}}
    child [thn] {node [circ] {}}
  }
};

\def\wx{0.2}
\def\wy{0.2}

\foreach \i in {r,a,b,c,d,e,f,L,R} {
  \foreach \sx in {+,-} {
    \foreach \sy in {+,-} {
      \path (\i)+(\sx\wx,\sy\wy) coordinate (\i\sx\sy);
    }
  }
}

\begin{pgfonlayer}{background}
  \pgfsetcornersarced{\pgfpoint{4pt}{4pt}}
  \fill [fill = Blu, thin]
    (a-+) --
    (L-+) -- (L--) --
    (R+-) -- (R++) --
    (a++) -- cycle;
  \fill [fill = Blk!50, opacity = 0.5]
    (r-+) -- (a-+) --
    (a--) -- (a+-) --
    (b--) -- (b+-) --
    (c--) -- (c+-) --
    (d--) -- (d+-) --
    (e--) -- (e+-) --
    (e++) -- (f++) --
    (r++) -- cycle;
\end{pgfonlayer}

\node [above left] at (a) {$v'$};
\path (T')+(5*\wx,-\wy) node {$\mc{T}'$};
\path (a)+(0,-3*\wy) node {$\mc{A}$};

\end{tikzpicture}
\caption{An example of a protocol tree $\mc{T}$ and its corresponding subtree $\mc{T}'$ (gray region).
We use black edges and empty nodes to indicate that zero progress is made at that point of the protocol. Purple edges and solid nodes indicate nonzero progress. Since the marked node $v'$ has a child with nonzero progress, we can $\ve$-interpolate the local measurement $\mc{A}$ at $v'$ (blue region) for some nonzero $\ve$.}
\label{fig:ITree}
\end{figure}

We claim that $\mu$ must be nonzero at some leaf of $\mc{T}$. This holds because $\mu(E) \ge \mu_0$ for all POVM elements $E$ of $\mc{M}_1$, $E$ is obtained by coarse graining measurement operators corresponding to the leaves, and $\mu$ is quasiconvex. By construction, $\mc{T}'$ has some vertex $v'$ with a child outside of $\mc{T}'$ with nonzero progress. 
Assume without loss of generality that Alice is the party performing a local measurement at $v'$ and denote that measurement by $\mc{A}$. In analogy to \Eq{IComp}, define a function $\mu'$ on Alice's space via
\begin{equation}
  \mu'(a) := \mu \bigl( \bigl( \sqrt{a'} a \sqrt{a'} \bigr) \x b' \bigr),
\end{equation}
where $a'\otimes b'$ is the POVM element that has been applied upon reaching node $v'$. Note that $\mu'$ is a valid progress function as it inherits all the properties required in \Def{IProgressFunction} from $\mu$ (e.g., $\mu'(I_A) = \mu(a' \x b') = 0$ by construction). Let $\lambda := \max_{a \in \mc{A}} \mu'(a)$ and note that $\lambda>0$ according to our assumption that $v'$ has children with nonzero progress.

Now, using \Thm{IKKB}, we can $\ve$-interpolate $\mc{A}$ with respect to $\mu'$ for any $\ve \in (0, \min\set{\lambda,\mu_0})\subseteq [0,\lambda]$. Any such $\ve$-interpolation of $\mc{A}$ with respect to $\mu'$ gives a product $\ve$-interpolation of $\mc{M}_1$ with respect to $\mu$, where the second-stage measurements complete the original LOCC protocol described by $\mc{T}$.
Since $\mathcal{M} = \mc{M}_2 \circ \mc{M}_1$, any product $\ve$-interpolation of $\mc{M}_1$ also gives a product $\ve$-interpolation of $\mc{M}$. Finally, applying \Lem{IProdToSep} yields an $\ve$-interpolation of $\mc{M}$ in~$\SEP$.
\end{proof}

To describe the consequences of \Thm{IMain}, let us first consider an example.

\begin{example*}
Let $\mc{M}$ be the von~Neumann measurement corresponding to the product basis shown in~\Fig{IDominoes}. Let $\mc{M}_{\LOCC} \in \LOCCN$ be a measurement implemented by the following two-step protocol (intuitively, it ``peels off'' the two extra tiles):
\begin{enumerate}
  \item Alice performs a two-outcome measurement $\set{I - \ketbra{3}, \ketbra{3}}$ and sends the outcome to Bob.
  \item If Alice got the first outcome, Bob applies the same measurement; otherwise he does nothing.
\end{enumerate}
\end{example*}

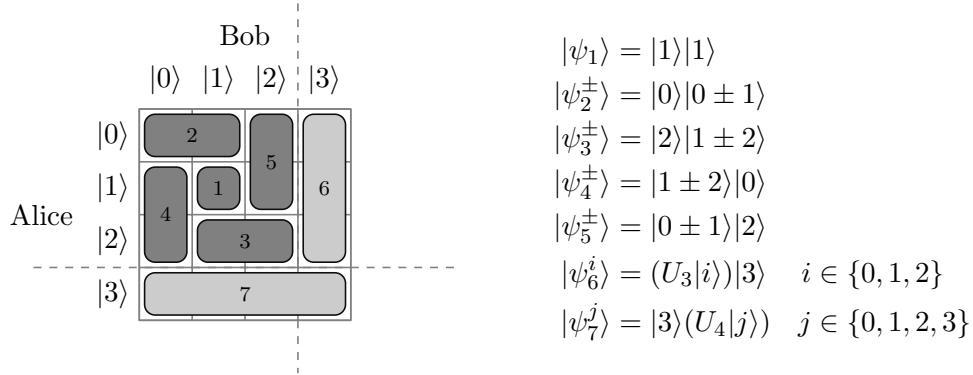
\begin{figure}[!ht]
\centering

\def\step{20pt} 

\begin{tikzpicture}[
  domino/.style = {rectangle, rounded corners = 0.2*\step, draw = black!95, fill = black!20},
  gridlines/.style = {gray, semithick}
]

  \newcommand{\tile}[4]{
    \pgfmathparse{#1+0+0.5*#3};
    \let\x = \pgfmathresult;
    \pgfmathparse{#2+0+0.5*#4};
    \let\y = \pgfmathresult;
    \drawtile{\x}{\y}{#3}{#4}
  }

  \newcommand{\drawtile}[4]{
    \node[semithick, domino,
      minimum width  = #3*\step - 0.2*\step,
      minimum height = #4*\step - 0.2*\step] at (#1*\step, #2*\step) {};
  }

  \newcommand{\tiling}[3]{
    \draw[step = \step, gridlines] (0,0) grid (#1*\step, #2*\step);
    #3
    \draw[gridlines] (0,0) -- (#1*\step, 0) -- (#1*\step, #2*\step) -- (0, #2*\step) -- cycle;
    \foreach \i in {0,...,3}{
      \node at (\i*\step+0.5*\step, 3.2)   {$|\i\rangle$};
      \node at (-0.35,-\i*\step+3.5*\step) {$|\i\rangle$};
    };
    \node at (-1.3,2*\step) {Alice};
    \node at (2*\step,5.4*\step) {Bob};
  }

  \tiling{4}{4}{
    \begin{scope}[domino/.style = {rectangle, rounded corners = 0.2*\step,
                                   draw = black!95, fill = black!50}]
      \tile{1}{2}{1}{1}
      \tile{0}{1}{1}{2}
      \tile{2}{2}{1}{2}
      \tile{0}{3}{2}{1}
      \tile{1}{1}{2}{1}
    \end{scope}
    \draw[gridlines, dashed] (-2*\step,\step) -- (6*\step,\step);
    \draw[gridlines, dashed] (3*\step,6*\step) -- (3*\step,-\step);
    \tile{3}{1}{1}{3}
    \tile{0}{0}{4}{1}
  }

  \node at (1.5*\step,2.5*\step) {\scriptsize$1$};
  \node at (1.0*\step,3.5*\step) {\scriptsize$2$};
  \node at (2.0*\step,1.5*\step) {\scriptsize$3$};
  \node at (0.5*\step,2.0*\step) {\scriptsize$4$};
  \node at (2.5*\step,3.0*\step) {\scriptsize$5$};
  \node at (3.5*\step,2.5*\step) {\scriptsize$6$};
  \node at (2.0*\step,0.5*\step) {\scriptsize$7$};

  \node at (12*\step,2.5*\step) {
    $\begin{aligned}
      \ket{\psi_1} &= \ket{1} \ket{1} \\
      \ket{\psi_2^\pm} &= \ket{0} \ket{0 \pm 1} \\
      \ket{\psi_3^\pm} &= \ket{2} \ket{1 \pm 2} \\
      \ket{\psi_4^\pm} &= \ket{1 \pm 2} \ket{0} \\
      \ket{\psi_5^\pm} &= \ket{0 \pm 1} \ket{2} \\
      \ket{\psi_6^i} &= (U_3 \ket{i}) \ket{3} & i &\in \{0,1,2\} \\
      \ket{\psi_7^j} &= \ket{3} (U_4 \ket{j}) & j &\in \{0,1,2,3\}
    \end{aligned}$
  };

\end{tikzpicture}
\caption{A product basis corresponding to domino states (dark gray) augmented with two extra tiles (light gray). A tile of size $l$ represents $l$ states that are supported only on that tile (see~\cite{Childs2012} for more details). All $16$ states are listed on the right, where $\ket{x \pm y} := (\ket{x}\pm\ket{y})/\sqrt2$. The light gray tiles are generated by unitaries $U_3$ and $U_4$ of size $3 \times 3$ and $4 \times 4$, respectively, that have no zero entries in the computational basis. For concreteness, $U_n$ could be the quantum Fourier transform modulo $n$.}
\label{fig:IDominoes}
\end{figure}

Note that $\mc{M}_{\LOCC}$ in this example is non-disturbing, so it can be completed by some measurement $\mc{M}' \in \SEP$ to obtain a decomposition $\mc{M} = \mc{M}' \circ \mc{M}_{\LOCC}$ as in \Thm{IMain}. We can specify $\mc{M}'$ more precisely by describing the measurement associated to each outcome of $\mc{M}_{\LOCC}$. If either of the parties obtains $\ketbra{3}$, they are left with one of the two long tiles and the protocol can be easily completed by a local measurement in an appropriate basis. Otherwise they are left with the problem of discriminating the domino states. Then no nontrivial non-disturbing local measurement is possible~\cite{Groisman, Walgate, Cohen07}, so Alice and Bob cannot proceed any further by using only LOCC. We call the remaining measurement \emph{purely separable} since it can be completed using separable operations, but no further progress can be made by LOCC without ruining the orthogonality of the states.

\section{Discussion} 
\label{sec:IDiscussion}

It is known that all LOCC measurements are separable but that some separable measurements are not in LOCC \cite{IBM, Childs2012}. Nevertheless, some separable measurements can be \emph{partially} implemented by LOCC. Purely separable measurements cannot even be partially implemented by LOCC. The resulting hierarchy is shown in \Fig{IDiagram}.

\begin{figure}[!ht]
\centering

\begin{tikzpicture}
\draw[fill=black!30] (0.3,-0.3) ellipse (58 pt and 32 pt);

\foreach \x/\name in {
    0/$\mathcal{M}_{\mathrm{LOCC}}$,
    1/$\mathcal{M}_{\mathrm{SEP}}\circ\mathcal{M}_{\mathrm{LOCC}}$,
    2/$\mathcal{M}_{\mathrm{SEP}}$}{
  \draw (\x*0.3,-\x*0.3) ellipse (30+28*\x pt and 15+17*\x pt);
  \draw (\x*0.3,-\x*0.9) node {\name};
}
\end{tikzpicture}
\caption{Subclasses of separable von~Neumann measurements. The innermost region corresponds to LOCC measurements. The shaded region corresponds to measurements that can be partially implemented by LOCC, \emph{i.e.}, decomposed as $\mc{M}_{\SEP} \circ \mc{M}_{\LOCC}$ as in \Thm{IMain}. The white region corresponds to purely separable measurements.}
\label{fig:IDiagram}
\end{figure}
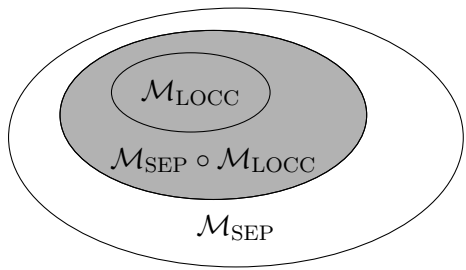

Our main result (\Thm{IMain}) characterizes purely separable measurements as precisely those for which $\ve$-interpolation is not possible for any positive $\ve$. We conclude that $\ve$-interpolatability for small $\ve > 0$ is the key feature that distinguishes $\LOCCN$ from purely separable von~Neumann measurements. In fact, this observation can be boosted to $\overline\LOCC$ (the closure of $\LOCC$), where $\mathcal{M}\in\overline{\LOCC}$ if there exists a sequence of measurements $\mc{M}_i \in \LOCCN$ that converge to $\mathcal{M}$ (see \cite{Eric12} for more details). This follows by combining \Thm{IMain} with the result of~\cite{KKB} that if $\mc{M}$ is a von~Neumann measurement, then $\mc{M} \in \LOCCN$ if and only if $\mc{M} \in \overline\LOCC$.

Our results suggest several open problems. One possible research direction is to extend our results beyond von~Neumann measurements. For example, can one generalize the notion of a progress function and prove an analogue of \Thm{IMain} for general POVMs or for the task of discriminating orthonormal states from an incomplete product basis?

Taking the idea of interpolation further, it could also be fruitful to find a continuous-time description of LOCC protocols. Such a description might give a new perspective on LOCC and a new tool for analyzing it. In particular, is it possible that the optimal protocol for some task is intrinsically continuous-time?

\section{Acknowledgements}

We thank David Roberson, Graeme Smith, and Andreas Winter for useful discussions, as well as the organizers of the workshop ``Operator structures in quantum information theory'' at the Banff International Research Station in 2012. This work was supported in part by NSERC, CRC, CFI, ORF, CIFAR, the Ontario Ministry of Research and Innovation, and the US ARO/DTO. MO acknowledges additional support from the DARPA QUEST program under contract number HR0011-09-C-0047.


\bibliographystyle{alphaurl}
\bibliography{Nonlocality}
\end{document}